\documentclass[12pt]{article}
\usepackage{latexsym,amsmath,amsfonts,natbib,bm}
\usepackage{amssymb,amsthm,rotating,threeparttable}
\usepackage{graphicx,subfigure}
\usepackage{amsmath}
\usepackage{setspace}
\usepackage{lscape}
\usepackage[pagewise,mathlines]{lineno}

\include{def}
\def\dispmuskip{\thinmuskip= 3mu plus 0mu minus 2mu \medmuskip=  4mu plus 2mu minus 2mu \thickmuskip=5mu plus 5mu minus 2mu}
\def\textmuskip{\thinmuskip= 0mu                    \medmuskip=  1mu plus 1mu minus 1mu \thickmuskip=2mu plus 3mu minus 1mu}
\def\beq{\dispmuskip\begin{equation}}    \def\eeq{\end{equation}\textmuskip}
\def\beqn{\dispmuskip\begin{displaymath}}\def\eeqn{\end{displaymath}\textmuskip}
\def\bea{\dispmuskip\begin{eqnarray}}    \def\eea{\end{eqnarray}\textmuskip}
\def\bean{\dispmuskip\begin{eqnarray*}}  \def\eean{\end{eqnarray*}\textmuskip}

\def\paradot#1{\vspace{1.3ex plus 0.7ex minus 0.5ex}\noindent{\bf\boldmath{#1.}}}
\newtheorem{theorem}{Theorem}

\newtheorem{lemma}{Lemma}
\newtheorem{proposition}{Proposition}

\newcommand*\samethanks[1][\value{footnote}]{\footnotemark[#1]}

\newcommand{\diag}{\text{diag}}

\newcommand{\wh}{\widehat}
\newcommand{\wt}{\widetilde}

\def\E{\mathbb{E}}                         % Expectation
\def\a{\alpha}

\def\s{\sigma}

\def\t{\theta}
\def\b{\beta}

\def\N{{\cal N}}

\def\ESS{\text{\rm ESS}}

\def\AIS{\text{\rm AIS}}
\def\AISEL{\text{\rm AISEL}}

\def\CV{\text{\rm CV}}
\def\CT{\text{\rm CT}}

\def\Var{\text{\rm Var}}

\def\IS{\text{\rm IS}}

\def\diag{\text{\rm diag}}
\def\opt{\text{\rm opt}}

\begin{document}
\doublespacing
%\pagewiselinenumbers
\title{Annealed Important Sampling for Models with Latent Variables}
\author{M.-N. Tran\thanks{Australian School of Business, University of New South Wales}
\and C. Strickland\samethanks
\and M. K. Pitt\thanks{Department of Economics, University of Warwick}
\and R. Kohn\samethanks[1]}
\maketitle

%======================================================================%
\begin{abstract}
This paper is concerned with Bayesian inference when the likelihood
is analytically intractable but can be unbiasedly estimated.
We propose an annealed importance sampling procedure
for estimating expectations with respect to the posterior.
The proposed algorithm is useful in cases
where finding a good proposal density is challenging,
and when estimates of the marginal likelihood are required.
The effect of likelihood estimation is investigated,
and the results provide guidelines on how to set up the precision of the likelihood estimation
in order to optimally implement the procedure.
The methodological results are empirically demonstrated in several simulated and real data examples.

\paradot{Keywords}
Intractable likelihood; Latent variables; Sequential Monte Carlo; Unbiasedness;
Marginal likelihood
\end{abstract}

\newpage
%======================================================================%
\section{Introduction} \label{Sec:introduction}
%======================================================================%
Models incorporating latent variables
 are very popular in many statistical applications.
For example, generalized linear mixed models \citep{Fitzmaurice:2011},
that use latent variables to account for dependence between observations,
appear in the genetics, social and medical sciences literatures
as well as many other areas of statistics.
State space models \citep{Durbin:2001},
whose latent variables follow a Markov process,
are used in economics, finance and engineering.
Gaussian process classifiers \citep[see, e.g.,][]{Filippone:2013,Rasmussen:2006},
that use a set of latent variables distributed
as a Gaussian process to account for uncertainty in predictions,
are used in computer science.

Inference about the model parameters $\theta$ in latent variable models can be  challenging because the likelihood
is expressed as an integral over the latent variables.
This integral is analytically intractable in general. It can  also be computationally challenging
when the dimension of the latent variables is high.
Recent work by \cite{Beaumont:2003,Andrieu:2009,Andrieu:2010}
shows that it is possible to carry out Bayesian inference
in latent variable models by using an unbiased estimate of the likelihood within a Markov chain Monte Carlo (MCMC)
sampling scheme. This method is known as particle MCMC.
However, the resulting Markov chain can often be trapped in local modes
and it is difficult to asses if it has converged.
We find that even for simple models it is difficult for the Markov chain to mix adequately
and the chain can take a long time to converge if the log of the estimated  likelihood is too variable.
The marginal likelihood is often used to choose between models. Another drawback of MCMC in general, and
particle MCMC in particular, is that it is often difficult to use it to estimate the marginal likelihood.

Another approach to Bayesian inference for models with latent variables
is importance sampling squared ($\IS^2$) proposed in \cite{Tran:2013}.
They show that importance sampling (IS)
with the likelihood replaced by its unbiased estimate is still valid for estimating
expectations with respect to the posterior.
$\IS^2$ offers several advantages over particle MCMC: (1)
It is easy to estimate the standard errors of the estimators;
(2) It is straightforward to parallelize the computation;
(3) It is straightforward to estimate the marginal likelihood.
However, as is typical of importance sampling algorithms,
a potential drawback with $\IS^2$ is that its performance may depend heavily on
the proposal density  for $\theta $. A good proposal density may
 be difficult to obtain in complex models.

When it is possible to evaluate the likelihood,
annealed importance sampling (AIS) \citep{Neal:2001} is a useful method for estimating
expectations with respect to the posterior for $\theta$.
AIS is an importance sampling method in which samples are first
drawn from an easily-generated distribution
and then moved towards the distribution of interest
through Markov kernels.
AIS explores the parameter space efficiently and is useful in cases where the target distribution is
multimodal and/or  when choosing an appropriate proposal density is challenging.

This article proposes an AIS algorithm for Bayesian inference when working with
an estimated likelihood, which we denote as AISEL (annealed importance sampling with an estimated likelihood).
Our first contribution is to show that the algorithm is valid for
estimating expectations with respect to the exact posterior when the likelihood is
estimated unbiasedly.
As with particle MCMC and $\IS^2$, it is important to understand the effect  of
estimating the likelihood on the resulting
inference.
The second contribution of this article is to answer this question
by comparing the efficiency of AISEL with the efficiency of the corresponding AIS procedure
that assumes that the 
likelihood is available. We show that the ratio of the efficiency of AISEL to that of AIS 
 is smaller than or equal to 1, and the ratio is equal to 1
if and only if the estimate of the likelihood is exact.
This  ratio  decreases exponentially with the
product of the variance of the log of the estimated likelihood
and a term that depends on the the annealing schedule
in the AIS algorithm.
The term based on the annealing schedule is small if the annealing schedule evolves slowly.
This result allows us to understand how much accuracy is lost when working with an estimated likelihood.
An attractive feature of AISEL is that it is more robust than $\IS^2$ and particle MCMC to the variability of the log likelihood
estimate. This is important when only highly variable estimates of the likelihood are available, which often occurs
if it is expensive to obtain accurate estimates of the likelihood.

The third contribution of the article is to provide theory and practical guidelines for optimally
choosing the number of particles to estimate the likelihood so as to minimize the overall
computational cost for a given precision.
The fourth contribution is to describe an efficient yet simple method
to compute the marginal likelihood, which is important for model choice.

The SMC$^2$  algorithm of \cite{Chopin:2012} sequentially updates
the posterior of the model parameters as new observations arrive.
The validity of the method is justified because
the likelihood estimated by the particle filter is unbiased.
In contrast, our AIS algorithm uses all the data and is static.
As discussed in their paper \citep[][Section 5.2]{Chopin:2012},
the reasoning used to justify SMC$^2$ does not apply to  tempered sampling
 in the spirit of the AIS,
because a tempered likelihood estimator is not an unbiased estimator of the corresponding tempered likelihood.
However, Section \ref{sec:formal justify} uses variable augmentation to
justify the validity of the AIS method when working with an unbiased likelihood estimate.

We illustrate the proposed methodology through a simulated example,
as well as the analysis of  a Pound/Dollar exchange rate dataset using a stochastic volatility model.
We show in these examples that the AIS method leads
to efficient inference when optimally implemented.

The article is organized as follows.
Section \ref{Sec:AIS} reviews the original AIS of \cite{Neal:2001}.
Section \ref{Sec:anneal} presents the main results.
Section \ref{Sec:Applications} presents the examples
and Section \ref{Sec:conclusion} concludes.
The technical proofs are in the Appendix.

%======================================================================%
\section{Annealed importance sampling}\label{Sec:AIS}
%======================================================================%
Let $p(y|\theta)$ be the density of the data $y$, where $\theta$ is a parameter vector belonging to a space $\Theta\subset\mathbb{R}^d$.  Let
$p(\t)$ be the prior for $\theta$ and $\pi(\t)\propto p(\t)p(y|\t)$  its posterior.
We are interested in the case where the likelihood $p(y|\theta)$ is analytically intractable but can be unbiasedly estimated.
The primary objective in Bayesian inference is to estimate an integral of the form
\beq\label{eq:target integral}
\E_\pi(\varphi)=\int_\Theta \varphi(\t)\pi(\theta)d\t \ ,
\eeq
for some $\pi$-integrable function $\varphi$ on $\Theta$.
We are also interested in estimating the marginal likelihood 
\beq\label{eq:llh}
p(y)=\int p(\t)p(y|\t)d\t.
\eeq

We now present the AIS procedure of \cite{Neal:2001}
when the likelihood $p(y|\theta)$ can be evaluated pointwise.
Let $\pi_0(\t)$ be some easily-generated density, such as a $t$ density or the prior.
Let $a_t,\ t=0,1,...,T$, be a sequence of real numbers such that $0=a_0<...<a_T=1$, which we call the annealing schedule.
A convenient choice is $a_t=t/T$.
AIS constructs the following  sequence of interpolation densities $\xi_{a_t}(\t),\ t=0,...,T$,
\beqn
\xi_{a_t}(\t) = \frac{\eta_{a_t}(\t)}{\int\eta_{a_t}(\t)d\t},\;\;\;\text{with}\;\;\;
\eta_{a_t}(\t) = \pi_0(\t)^{1-a_t}[p(\t)p(y|\t)]^{a_t}.
\eeqn
Note that $\xi_{a_0}(\t)=\pi_0(\t)$ and $\xi_{a_T}(\t)$ is the posterior $\pi(\t)$ of interest.
Denote by $K_{\xi_{a_t}}(\t,\cdot)$ a Markov kernel density conditional on $\t$ with invariant distribution $\xi_{a_t}$, $t=1,...,T-1$.
AIS draws $M$ weighted samples $\{w_i,\theta_i\}_{i=1}^M$ as follows.

\paradot{AIS algorithm} For $i=1,...,M$
\begin{itemize}
\item Generate $\theta^{(1)}\sim \xi_{a_0}(\cdot)$.
\item For $t=1,...,T-1$, generate $\theta^{(t+1)}\sim K_{\xi_{a_t}}(\theta^{(t)},\cdot)$.
\item Set $\t_i=\t^{(T)}$ and compute the unnormalized weight
\beqn
w_i = \frac{\eta_{a_1}(\theta^{(1)})}{\eta_{a_0}(\theta^{(1)})}\times\frac{\eta_{a_2}(\theta^{(2)})}{\eta_{a_1}(\theta^{(2)})}\times...\times \frac{\eta_{a_T}(\theta^{(T)})}{\eta_{a_{T-1}}(\theta^{(T)})}.
\eeqn
\end{itemize}
Note that the algorithm is parallelizable.
\cite{Neal:2001} shows that the above algorithm is an IS procedure operating on the extended space $\Theta^T$
with the artificial target density proportional to
\beqn
f(\t^{(1)},...,\t^{(T)})=\eta_{a_T}(\t^{(T)})L_{\xi_{a_{T-1}}}(\t^{(T)},\t^{(T-1)})...L_{\xi_{a_1}}(\t^{(2)},\t^{(1)}),
\eeqn
where
\beqn
L_{\xi_{a_t}}(\t^{(t+1)},\t^{(t)})=\frac{\xi_{a_t}(\t^{(t)})K_{\xi_{a_t}}(\t^{(t)},\t^{(t+1)})}{\xi_{a_t}(\t^{(t+1)})}
=\frac{\eta_{a_t}(\t^{(t)})K_{\xi_{a_t}}(\t^{(t)},\t^{(t+1)})}{\eta_{a_t}(\t^{(t+1)})}
\eeqn is a backward kernel density,
and the proposal density
\beqn
g(\t^{(1)},...,\t^{(T)})=\xi_{a_0}(\t^{(1)})K_{\xi_{a_1}}(\t^{(1)},\t^{(2)})...K_{\xi_{a_{T-1}}}(\t^{(T-1)},\t^{(T)}).
\eeqn
The original target density $\pi=\xi_{a_T}$ is the last marginal of $f(\t^{(1)},...,\t^{(T)})$ because
\beqn
\int L_{\xi_{a_t}}(\t^{(t+1)},\t^{(t)})d\t^{(t)}=1
\eeqn 
for all $t$.
This shows that AIS is a valid IS method with the weight $w\propto f/g$.
Hence, the weighted samples $\{W_i,\t_i\}_{i=1}^M$ with $W_i=w_i/\sum_{j=1}^M w_j$ approximate $\pi(\t)$,
i.e., $\sum_{i=1}^M W_i\varphi(\t_i)\to \E_\pi(\varphi)$ almost surely, for any $\pi$-integrable function $\varphi$.

The AIS procedure explores the parameter space efficiently,
and is useful when the target distribution is multimodal
and when choosing an appropriate proposal density is challenging.

%======================================================================%
\section{Annealed importance sampling with an estimated likelihood} \label{Sec:anneal}
%======================================================================%
This section presents an AISEL algorithm for
estimating the integral \eqref{eq:target integral} when the likelihood is
analytically intractable but can be estimated unbiasedly.
Let $\wh p_N(y|\t)$ denote an estimator of $p(y|\t)$ with $N$ the number of particles used to estimate the likelihood.
We define a sequence of functions $\wt\eta_{a_t}(\t),\ t=0,...,T$, by
\beqn
\wt\eta_{a_t}(\t) = \pi_0(\t)^{1-a_t}[p(\t)\wh p_N(y|\t)]^{a_t}.
\eeqn
We propose the following algorithm for generating $M$ weighted samples $\{\wt W_i,\t_i\}_{i=1}^M$
which approximate the posterior $\pi(\t)$; see Section~\ref{sec:formal justify}

\paradot{Algorithm 1 (AISEL)} Let $\pi_0(\t)$ be some easily-generated density.
\begin{enumerate}
\item Generate $\t_i\sim\pi_0(\t),\ i=1,...,M$.
Set $\wt W_i=1/M,\ i=1,...,M$.
\item For $t=1,...,T$
\begin{itemize}
\item[(i)] Weighting: compute the unnormalized weights
\beq\label{e:weights}
\wt w_i=\wt W_i\frac{\wt\eta_{a_t}(\t_i)}{\wt\eta_{a_{t-1}}(\t_i)}=\wt W_i[\pi_0(\t_i)]^{a_{t-1}-a_t}[p(\t_i)\wh p_N(y|\t_i)]^{a_t-a_{t-1}},
\eeq
and set the new normalized weights $\wt W_i={\wt w_i}/{\sum_{j=1}^{M}\wt w_j}$.
\item[(ii)] Resampling: If $\text{ESS}=1/\sum_{i=1}^{M} \wt W_i^2<\a M$ for some $0<\a<1$, e.g. $\a=1/2$, then resample from $\{\wt W_i,\t_i\}_{i=1}^{M}$, set $\wt W_i=1/M$ and (still) denote the resamples by $\{\wt W_i,\t_i\}_{i=1}^{M}$.
\item[(iii)] Markov move: for each $i=1,...,M$, move the sample $\t_i$ according to a Metropolis-Hastings step as follows.
Let $q_t(\t|\t^c)$ be a proposal with $\t^c=\t_i$ the current state.
Generate $\t^p\sim q_t(\t|\t^c)$ and set $\t_i=\t^p$ with probability
\beqn
\min\left(1,\frac{\pi_0(\t^p)^{1-a_t}[p(\t^p)\wh p_N(y|\t^p)]^{a_t}q_t(\t^c|\t^p)}{\pi_0(\t^c)^{1-a_t}[p(\t^c)\wh p_N(y|\t^c)]^{a_t}q_t(\t^p|\t^c)}\right).
\eeqn
Otherwise set $\theta_i = \theta^c$.
\end{itemize}
\end{enumerate}

If $\wh p_N(y|\t)=p(y|\t)$,
the above algorithm is a special case of the SMC sampler in \cite{DelMoral:2006}
for sampling from the sequence of distributions $\xi_{a_t}$.
The AIS algorithm of \cite{Neal:2001} is a special case
of this algorithm in which no resampling steps are performed.
It is widely known in the literature that
it is beneficial to incorporate resampling steps \citep{DelMoral:2006}.
The algorithm is also closely related to the resample-move algorithm of \cite{Gilks:2001},
except that they perform the resampling step in every iteration $t$.
%======================================================================%
\subsection{Formal justification}\label{sec:formal justify}
%======================================================================%
The output of Algorithm 1 is weighted samples $\{\wt W_i,\t_i\}_{i=1}^{M}$.
To prove that this algorithm is valid,
i.e. $\{\wt W_i,\t_i\}_{i=1}^{M}$ approximate $\pi(\t)$,
we make the following assumption.

\paradot{Assumption 1} $\E[\wh p_N(y|\t)]=p(y|\t)$ for every $\theta\in\Theta$.

Let us write $\wh p_N(y|\t)$ as $p(y|\t)e^z$ where $z=\log\;\wh p_N(y|\t)-\log\;p(y|\t)$ is a random variable
whose distribution is governed by the randomness occurring when estimating the likelihood $p(y|\t)$.
Let $g_N(z|\t)$ be the density of $z$. Assumption 1 implies that
\beqn
\E(e^z)=\int_\mathbb{R} e^zg_N(z|\t)dz=1.
\eeqn
We define
\beq\label{e:eobietgi1}
\pi_N(\t,z)=p(\t)g_N(z|\t)p(y|\t)e^z/p(y)
\eeq
as the joint density of $\t$ and $z$ on the extended space $\widetilde\Theta=\Theta\otimes\Bbb{R}$.
Then its first marginal is the posterior $\pi(\t)$ of interest, i.e.
\beq\label{e:key-relation}
\int_\mathbb{R}\pi_N(\t,z)dz=\pi(\t).
\eeq
We define the following sequence of interpolation densities on $\wt\Theta=\Theta\otimes\Bbb{R}$
\beq\label{e:eobietgi2}
\wt\xi_{a_t}(\t,z)=\frac{\wt\eta_{a_t}(\t,z)}{\int \wt\eta_{a_t}(\t,z)d\t dz},
\eeq
with
\beqn
\wt\eta_{a_t}(\t,z) = \pi_0(\t)^{1-a_t}[p(\t)p(y|\t)e^z]^{a_t}g_N(z|\t),\;\;t=0,...,T.
\eeqn
Note that $\wt\xi_{a_T}(\t,z)=\pi_N(\t,z)$, which is our new target density defined on $\wt\Theta$.
Algorithm 1 is entirely equivalent to the following procedure.

\paradot{Algorithm 1'}
\begin{enumerate}
\item Generate $(\t_i,z_i)\sim\wt\xi_{a_0}(\t,z)=\pi_0(\t)g_N(z|\t)$, i.e. generate $\t_i\sim\pi_0(\t)$ then $z_i\sim g_N(z|\t_i),\ i=1,...,M$.
Set $\wt W_i=1/M$.
\item For $t=1,...,T$
\begin{itemize}
\item[(i)] Weighting: compute the weights
\beqn
\wt w_i=\wt W_i\frac{\wt\eta_{a_t}(\t_i,z_i)}{\wt\eta_{a_{t-1}}(\t_i,z_i)}=\wt W_i[\pi_0(\t_i)]^{a_{t-1}-a_t}[p(\t_i)p(y|\t_i)e^{z_i}]^{a_t-a_{t-1}}.
\eeqn
and set the new normalized weights $\wt W_i={\wt w_i}/{\sum_{j=1}^{M}\wt w_j}$.
\item[(ii)] Resampling: If $\text{ESS}<\alpha M$, then resample from $\{\wt W_i,\t_i,z_i\}_{i=1}^{M}$, set $\wt W_i=1/M$ and denote the resamples by $\{\wt W_i,\t_i,z_i\}_{i=1}^{M}$.
\item[(iii)] Markov move: for each $i=1,...,M$, move the sample
$(\t^c,z^c)=(\t_i,z_i)$ by a Metropolis-Hastings kernel $K_{\wt\xi_{a_t}}(\cdot,\cdot)$
as follows. Generate a proposal $(\t^p,z^p)$ from the proposal density $\wt q_t(\t^p,z^p|\t^c,z^c)=q_t(\t^p|\t^c)g_N(z^p|\t_p)$.
Set $(\t_i,z_i)=(\t^p,z^p)$ with probability
\bean
\text{prob} &=& \min\left(1,\frac{\wt\eta_{a_t}(\t^p,z^p)\wt q_t(\t^c,z^c|\t^p,z^p)}{\wt\eta_{a_t}(\t^c,z^c)\wt q_t(\t^p,z^p|\t^c,z^c)}\right)\\
&=&\min\left(1,\frac{\pi_0(\t^p)^{1-a_t}[p(\t^p)\wh p_N(y|\t^p)]^{a_t}q_t(\t^c|\t^p)}{\pi_0(\t^c)^{1-a_t}[p(\t^c)\wh p_N(y|\t^c)]^{a_t}q_t(\t^p|\t^c)}\right).
\eean
\end{itemize}
\end{enumerate}
Phrased differently, generating weighted samples $\{\wt W_i,\t_i\}_{i=1}^{M}$ according to Algorithm 1 is equivalent to generating
weighted samples $\{\wt W_i,\t_i,z_i\}_{i=1}^{M}$ according to Algorithm 1'.
Algorithm 1' is exactly the SMC sampler of \cite{DelMoral:2006}
for sampling from the sequence $\wt\xi_{a_t}(\t,z),\ t=0,...,T$,
in which the backward kernel used is the backward kernel (30) in their paper.
Therefore, the weighted samples $\{\wt W_i,\t_i,z_i\}_{i=1}^{M}$ produced after the last iteration $T$
approximate $\wt\xi_{a_T}(\t,z)=\pi_N(\t,z)$, i.e.
\beqn
\sum_{i=1}^M \wt W_i\wt\varphi(\t_i,z_i)\stackrel{a.s.}{\longrightarrow} \int \wt\varphi(\t,z)\pi_N(\t,z)dzd\t,\;\;M\to\infty,
\eeqn
for any $\pi_N$-integrable function $\wt\varphi(\t,z)$ on $\wt\Theta$.
Given the function $\varphi(\t)$ in \eqref{eq:target integral},
we define the corresponding function $\wt\varphi$ on $\wt\Theta$ by $\wt\varphi(\t,z)=\varphi(\t)$,
then

\bean
\wh\varphi_\AISEL=\sum_{i=1}^M \wt W_i\varphi(\t_i)=\sum_{i=1}^M \wt W_i\wt\varphi(\t_i,z_i)&\stackrel{a.s.}{\longrightarrow}& \int_{\wt\Theta} \wt\varphi(\t,z)\pi_N(\t,z)dzd\t\\
&=&\int_{\Theta} \varphi(\t)\pi(\t)d\t\\
& =& \E_\pi(\varphi),
\eean
as $M\to\infty$.
This justifies Algorithm 1.
We refer to $\wh\varphi_\AISEL$ as the AISEL estimator of $\E_\pi(\varphi)$.

\iffalse
\paradot{Remark 1} If the primary goal is to generate samples
that are approximately distributed as $\pi(\t)$,
then at the last iteration $T$ of Algorithm 1, the resampling step must be performed.
If the primary goal is to estimate the expectation $\E_\pi(\varphi)$,
then at the last iteration $T$, then the  resampling step should not be performed.
This is because resampling adds more ``noise" to the estimator as shown by \cite{Chopin:2004}.
\fi

\paradot{Remark 1} It is advisable to perform the Markov move step over a few burn-in iterations
so that the samples move closer to the equilibrium distribution.

\paradot{Remark 2}
If $\varphi$ satisfies the conditions in Theorem 1 of \cite{Chopin:2004}, then
\beq
\sqrt{M}\Big(\wh\varphi_\AISEL-\E_\pi(\varphi)\Big)\stackrel{d}{\to}\N(0,\sigma^2_\AISEL(\varphi)) \;\;\text{as}\;\;M\to\infty,
\eeq
with the asymptotic variance $\sigma^2_\AISEL(\varphi)$ defined recursively as in \cite{Chopin:2004}.

Except for the special case in which
no resampling steps in Algorithm 1 are performed,
the asymptotic variance $\sigma^2_\AISEL(\varphi)$,
and therefore the variance of $\wh\varphi_\AISEL$,  does not admit a closed form.
A natural and potential technique to estimate $\Var(\wh\varphi_\AISEL)$
is to run Algorithm 1 in batches independently and in parallel.
Then, we have several independent batches of weighted samples $\{\wt W_i^{(r)},\t_i^{(r)}\}_{i=1}^{M_r}$, $r=1,...,R$ with $\sum_r M_r=M$,
and the corresponding $R$ independent estimates $\wh\varphi^{(r)}_\AISEL$ of $\E_\pi(\varphi)$.
The variance of the estimator $\wh\varphi_\AISEL$ can be estimated by
\beqn
\wh\Var(\wh\varphi_\AISEL)=\frac{1}{R}\sum_{r=1}^R(\wh\varphi^{(r)}_\AISEL-\overline{\wh\varphi})^2\;\;\text{with}\;\;\overline{\wh\varphi}=\frac{1}{R}\sum_{r=1}^R\wh\varphi^{(r)}_\AISEL.
\eeqn
If no resampling steps are performed, then the particles $\t_i$ are independent
and we have a closed form expression for estimating the asymptotic variance of $\wh\varphi_\AISEL$
\beq\label{e:ISvar_est1}
\wh{\sigma^2_{\AISEL}(\varphi)}={M\sum_{i=1}^{M}\big(\varphi(\t_i)-\wh\varphi_{\AISEL}\big)^2\wt W_i^2}.
\eeq
It is straightforward to show that this estimate is consistent.
However, it is important to perform resampling if necessary.

\iffalse
\paradot{Remark 4} In the examples in this article,
we assume that $T$ is even and let $T_1=T/2$.
We set $a_t=0.2t/T_1$ for $t=0,...,T_1$
and $a_t=0.8(t-T_1)/T_1$ for $t=T_1+1,...,T$.
This setting allows an initially smooth transition
between the interpolation distributions, which is useful in cases the initial distribution $\pi_0$
is very different from the target $\pi$.
\fi

%======================================================================%
\subsection{Estimating the marginal likelihood}\label{Subsec:marginal likelihood}
%======================================================================%
Marginal likelihood \eqref{eq:llh}
is important for model comparison purposes.
Except for some trivial cases,
computing the marginal likelihood is challenging because of its integral form.
\cite{Friel:2008} propose a very efficient method,
called power posterior method, for estimating the marginal likelihood, that exploits the temparing sampling framework as in AIS.
This section extends the power posterior method to the case with latent variables. 

We consider for now a {\it continuous} sequence of interpolation densities \eqref{e:eobietgi2} as follows
\beqn
\wt\xi_s(\t,z)=\frac{\wt\eta_s(\t,z)}{\int \wt\eta_s(\t,z)d\t dz},\quad\text{with}\quad
\wt\eta_s(\t,z) = \pi_0(\t)^{1-s}[p(\t)p(y|\t)e^z]^{s}g_N(z|\t),\ 1\leq s\leq 1.
\eeqn
We have the following result. The proof is in the Appendix.
\begin{proposition}\label{proposition} Under Assumption 1, the log of the marginal likelihood $\log\;p(y)$ can be expressed as 
\beq\label{eq:marginal llh}
\log p(y) = \int_0^1\E_{(\t,z)\sim\wt\xi_s}\left[\log\frac{p(\t)p(y|\t)e^z}{\pi_0(\t)}\right]ds.
\eeq
\end{proposition}
Let $f(s)$ denote the integrand in the right side of \eqref{eq:marginal llh}.
The scalar integral in \eqref{eq:marginal llh} can be deterministically approximated by
\beq\label{eq:hat_log_py}
\widehat{\log p(y)}=\sum_{t=0}^{T-1}(a_{t+1}-a_t)\frac{f(a_{t+1})+f(a_{t})}{2},
\eeq
with $\{a_t,\ t=0,...,T\}$ the annealing schedule as in Section \ref{Sec:AIS}.
The function $f(a_t)$ can be estimated by
\beq\label{eq:f_a_t}
\wh f(a_t)=\sum_{i=1}^M\wt W_i^{(t)}\log\frac{p(\t_i^{(t)})\wh p_N(y|\t_i^{(t)})}{\pi_0(\t_i^{(t)})}
\eeq
with $\{\wt W_i^{(t)},\t_i^{(t)}\}_{i=1}^M$ the output of Algorithm 1 after iteration $t$.

This approach of estimating the marginal likelihood
fits naturally to the AISEL procedure and is straightforward to implement.
Note that the values $\wh p_N(y|\t_i^{(t)})$ in \eqref{eq:f_a_t} 
can be used for the weighting step in iteration $t+1$ of Algorithm 1,
hence no extra computation is needed except calculation in \eqref{eq:hat_log_py} and \eqref{eq:f_a_t}.

%======================================================================%
\subsection{The effect of estimating the likelihood}\label{Subsec:2.1}
%======================================================================%
The efficiency of an IS procedure with proposal density $g$ and weights $w_i$ is often measured by
the effective sample size defined by (see, e.g. \cite{Neal:2001} and \citet[Chapter 2]{Liu:2001})
\beqn
\ESS = \frac{M}{1+\Var_{g}(w_i/\E_{g}[w_i])}=\frac{M}{1+\CV_g(w_i)},
\eeqn
where $\CV_g(w_i)=\Var_{g}(w_i)/(\E_{g}[w_i])^2$ is often called the coefficient of variation of the unnormalized weights $w_i$.
The bigger the ESS the more efficient the IS procedure.
This section investigates how much the ESS is reduced when working with an estimated likelihood.

We consider the case of the original AIS procedure, i.e. Algorithm 1 without the resampling step,
and work with the notation in Algorithm 1'.
Write $x=(\t,z)$ and $x^{(t)}=(\t^{(t)},z^{(t)})$. Without the resampling step, Algorithm 1' can be written as
\begin{itemize}
\item Generate $x^{(1)}=(\t^{(1)},z^{(1)})\sim\wt\xi_{a_0}(\t,z)=\pi_0(\t)g_N(z|\t)$.
\item For $t=1,...,T-1$, generate $x^{(t+1)}=(\t^{(t+1)},z^{(t+1)})$ from the Markov kernel $K_{\wt\xi_{a_t}}(x^{(t)},\cdot)$.
\item Set $(\t_i,z_i)=(\t^{(T)},z^{(T)})$ and compute the corresponding unnormalized weight
\beqn
\wt w_i= \frac{\wt\eta_{a_1}(x^{(1)})}{\wt\eta_{a_0}(x^{(1)})}\times \frac{\wt\eta_{a_2}(x^{(2)})}{\wt\eta_{a_1}(x^{(2)})}\times...\times
\frac{\wt\eta_{a_T}(x^{(T)})}{\wt\eta_{a_{T-1}}(x^{(T)})}.
\eeqn
\end{itemize}
Denote by $\ESS_\AIS$ and $\ESS_\AISEL$ the effective sample sizes of the AIS procedures when the likelihood is given and when it is estimated, respectively.

We make the following assumption which is satisfied in almost cases.

\paradot{Assumption 2}
There exists a function $\lambda(\t)$ such that for each $\t\in\Theta$, $\lambda^2(\t)<\infty$ and
\beqn
\sqrt{N}\left(\wh p_N(y|\t)-p(y|\t)\right)\stackrel{d}{\to}\N(0,\lambda^2(\t))\;\;\text{as}\;\;N\to\infty.
\eeqn
The following lemma follows immediately from Assumption 2 using the second order $\delta$-method.
\begin{lemma}\label{l:CLT} Let $\gamma^2(\t)=\lambda^2(\t)/p(y|\t)^2$,
and suppose that Assumption 2 holds. Let $z\sim g_N(z|\t)$. Then
\beqn
\sqrt{N}\left(\frac{z+\frac{\gamma^2(\t)}{2N}}{\gamma(\t)}\right)\stackrel{d}{\to}\N(0,1) \;\;\text{as}\;\;N\to\infty.
\eeqn
\end{lemma}

Following \cite{Pitt:2012}, we make the following further assumptions.

\paradot{Assumption 3} (i) The density $g_N(z|\t)$ of $z$ is $\N(-\frac{\gamma^2(\t)}{2N},\frac{\gamma^2(\t)}{N})$.
(ii) For a given $\sigma^2>0$, let $N$ be a function of $\t$ and $\s^2$ such that
$\text{Var}(z)\equiv\s^2$, i.e. $N=N_{\s^2}(\t)=\gamma^2(\t)/\s^2$.

Assumption 3(i) is justified by Lemma \ref{l:CLT}.
Assumption 3(ii) keeps the variance $\Var(z)$ constant across different values of $\t$, thus
making it easy to associate the ESS with $\s$.
Under Assumption 3, the density $g_N(z|\t)$ depends only on $\s$
and is denoted by $g(z|\s)$.

\paradot{Assumption 4} $K_{\xi_{a_t}}(\t,\cdot)=\xi_{a_t}(\cdot)$ and $K_{\wt\xi_{a_t}}(x,\cdot)=\wt\xi_{a_t}(\cdot)$.

As in \cite{Neal:2001}, this assumption
separates out the effect related to Markov chain convergence
and allows us to study the effect of estimating the likelihood on the sequential sampling scheme.

\begin{theorem}\label{the:theorem ESS}
Suppose that Assumptions 1-4 hold.
Then
\beq\label{eq:ESS}
\frac{\ESS_\AISEL}{\ESS_\AIS}=\exp\left(-\tau\s^2\right),
\eeq
with $\tau=\sum_{t=1}^T(a_t-a_{t-1})(2a_t-1)>0$ for any sequence $0=a_0<a_1<...<a_T=1$.
\end{theorem}
The theorem, whose proof is in the Appendix, shows that the efficiency is reduced by the factor $\exp(\tau\s^2)$ when working with an estimated likelihood.
If $a_t=1/T$, then $\tau=1/T$ for all $t$, and the theorem shows that increasing $T$ and thus making  the $\wt\xi_{a_t}$
closer to each other helps improve efficiency.

In the $\IS^2$ approach, \cite{Tran:2013} show that
$\ESS_{\IS^2}/\ESS_\IS=\exp(-\s^2)$,
where $\ESS_{\IS^2}$ and $\ESS_\IS$ are the effective sample sizes
of IS when the likelihood is estimated and given, respectively.
Similarly in particle MCMC, \cite{Pitt:2012} show that efficiency is reduced by a factor of  approximately $\exp(-\s^2)$
when working with an estimated likelihood.
We can see that
the accuracy of likelihood estimation is less important in AISEL 
than in $\IS^2$ and particle MCMC, because the factor $\exp(-\tau\s^2)$
can be made small when $\tau$ is decreased.
This means that AISEL can be more robust than $\IS^2$ and particle MCMC in cases
where we only have a rough estimate of the likelihood, or it is expensive to obtain an accurate estimate of the likelihood.

\subsection{Practical guidelines on selecting the number of particles}
This section studies how to select the number of particles $N$ optimally.
A large number of particles $N$ results in a precise likelihood estimate,
and therefore an accurate estimate of $\E_\pi(\varphi)$, but at a greater computational cost.
A small $N$ leads to a large variance of the likelihood estimator,
so we need a larger number of importance samples $M$ in order to obtain the desired accuracy of the AISEL estimator.
In either case, the computation is expensive.
It is important to select an optimal value of $N$ that minimizes the computational cost.

The time to compute the likelihood estimate $\wh p_N(y|\t)$ can be written
as $\tau_0+N(\t)\tau_1$ where $\tau_0\geq 0$ and $\tau_1>0$ \citep{Tran:2013}.
For example, if $\wh p_N(y|\t)$ is estimated by IS,
then $\tau_0$ is the overhead cost spent on estimating the proposal density
and $\tau_1$ is the computing time used to generate each sample and compute the weight.
Note that under Assumption 3, $N$ depends on $\t$ as $N=N_{\s^2}(\t)=\gamma^2(\t)/\s^2$.

The variance of the AISEL estimator is approximated as
\beqn
\Var(\wh\varphi_\AISEL)\approx\frac{\Var_\pi(\varphi)}{\ESS_\AISEL},
\eeqn
with $\Var_\pi(\varphi)=\E_\pi(\varphi-\E_\pi(\varphi))^2$.
From \eqref{eq:ESS},
\beq\label{eq:Var_AIS}
\Var(\wh\varphi_\AISEL)\approx\frac1M\Var_\pi(\varphi)(1+\CV_g(w_i))\exp(\tau\s^2).
\eeq
Let $P^*$ be a prespecified precision. Then we need approximately
\beqn
M(P^*) = \frac1{P^*}\Var_\pi(\varphi)(1+\CV_g(w_i))\exp(\tau\s^2)
\eeqn
particles in order to have that precision.
The required computing time to run AISEL is
\beq\label{eq:time}
\sum_{t=1}^T\sum_{i=1}^M(N(\t_i^{(t)})\tau_1+\tau_0)\approx TM\left(\frac{\bar\gamma^2}{\s^2}\tau_1+\tau_0\right)=
\frac{T}{P^*}\Var_\pi(\varphi)(1+\CV_g(w_i))\exp(\tau\s^2)\left(\frac{\bar\gamma^2}{\s^2}\tau_1+\tau_0\right),
\eeq
in which
\beqn
\frac{1}{TM}\sum_{t=1}^T\sum_{i=1}^M\gamma^2(\t_i^{(t)})\longrightarrow\bar\gamma^2=\frac1T\sum_{t=1}^T\E_{\xi_{a_t}}[\gamma^2(\t)],\;\; M\to\infty.
\eeqn
Therefore
\beq\label{eq:CT*}
\CT^*(\s^2)=\exp(\tau\s^2)\times \left(\frac{\bar\gamma^2}{\s^2}\tau_1+\tau_0\right)
\eeq
characterizes the computing as a function of $\s^2$,
which is minimized at
\beq\label{eq:sigsq_opt}
\s^2_\text{opt} = \begin{cases}
\frac{\sqrt{(\bar\gamma^2\tau\tau_1)^2+4\bar\gamma^2\tau\tau_0\tau_1}-\bar\gamma^2\tau\tau_1}{2\tau\tau_0},&\;\;\tau_0>0\\
1/\tau,&\;\;\tau_0=0.
\end{cases}
\eeq
The optimal number of particles $N$ is such that
\beqn
\Var_{N,\t}(z)=\Var(\log\wh p_{N}(y|\t))=\s^2_\text{opt}.
\eeqn

Let $\wh\Var_{N,\t}(z)$ be an estimate of $\Var_{N,\t}(z)$,
which can be obtained by using, e.g., the delta method
or the jackknife. See \cite{Tran:2013} for more details.
We suggest the following practical guidelines for tuning the optimal number of particles $N$.
Note that $N$ generally depends on $\t$ but this dependence is suppressed for notational simplicity.

\noindent{\bf The case $\tau_0=0$}. From \eqref{eq:CT*}, $\s^2_\text{opt}=1/\tau$. It is necessary to tune $N$ such that $\wh\Var_{N,\t}(z)=1/\tau$.
A simple strategy is to start with some small $N$ and increase it if $\wh\Var_{N,\t}(z)>1/\tau$.

\noindent{\bf The case $\tau_0>0$}. First, we need to estimate $\bar\gamma^2$.
Let $\{\t_1,...,\t_J\}$ be a few initial draws from the initial density $\xi_0(\t)$.
Then, start with some large $N_0$, $\bar\gamma^2$ can be initially estimated by
\beq\label{eq:gamma^2}
\wh{\bar\gamma}^2 = \frac{1}{J}\sum_{j=1}^J\wh\gamma^2(\t_j)=\frac{N_0}{J}\sum_{j=1}^J\wh\Var_{N_0,\t_j}(z),
\eeq
as $\wh\Var_{N_0,\t_j}(z)=\wh\gamma^2(\t_j)/N_0$.
By substituting this estimate of ${\bar\gamma}^2$ into \eqref{eq:sigsq_opt} we obtain an estimate $\wh{\s}^2_\text{opt}$ of ${\s}^2_\text{opt}$.
We now can start Algorithm 1 and update $\wh{\bar\gamma}^2$ (and therefore $\wh{\s}^2_\text{opt}$) as we go.
For each draw of $\t$, we start with some small $N$ and increase $N$ if $\wh\Var_{N,\t}(z)>\wh{\s}^2_\text{opt}$.

\paradot{Time normalized variance}
In the examples in Section \ref{Sec:Applications}
we use the time normalized variance (TNV) as a measure of efficiency \citep{Tran:2013} of a sampling procedure.
The TNV of the AISEL estimator $\wh\varphi_\AISEL$ is defined as
\beq\label{eq:TNV}
\text{TNV}(M,N) = \Var(\wh\varphi_\AISEL)\times\tau(M,N),
\eeq
where $\tau(M,N)$ is the total CPU time used to run the AISEL procedure with $M$ importance samples and $N$ particles.
From \eqref{eq:Var_AIS} and \eqref{eq:time},
\bean
\text{TNV}(M,N) &\approx& T\Var_\pi(\varphi)(1+\CV_g(w_i))\exp(\tau\s^2)\left(\frac{\bar\gamma^2}{\s^2}\tau_1+\tau_0\right)
\eean
and is propotional to $\CT^*(\s^2)$.

\paradot{Remark 3} Letting the optimal $N_\opt=N_\opt(\theta)$ depend on $\theta$
is theoretically interesting but might in some cases be ultimately inefficient.
The reason is that extra computing time is needed to tune $N$ for each $\t$.
A simple strategy is to make approximation that $\gamma^2(\t)\approx {\bar\gamma}^2$.
Then, the optimal number of particles,
which is constant across $\t$, is determined as
\beq\label{eq:N_opt}
N_\text{opt} = \begin{cases}
\frac{2\tau\tau_0}{\sqrt{(\tau\tau_1)^2+4\tau\tau_0\tau_1/\bar\gamma^2}-\bar\gamma^2\tau\tau_1},&\;\;\tau_0>0\\
\tau \bar\gamma^2,&\;\;\tau_0=0.
\end{cases}
\eeq
We follow this strategy in the examples below.

%======================================================================%
\section{Examples}\label{Sec:Applications}
%======================================================================%
\subsection{A simulation example}\label{sec:simul ex}
We generate a dataset from a mixed logistic regression model
\beqn
P(y_{ij}=1|\b,\eta_i) = \frac{\exp(\b_0+x_{ij}'\beta+\eta_{i0}+z_{ij}\eta_{i1})}{1+\exp(\b_0+x_{ij}'\beta+\eta_{i0}+z_{ij}\eta_{i1})},\;\;j=1,...,n_i,\; i=1,...,m
\eeqn
in which the random effects $\eta_i=(\eta_{i0},\eta_{i1})'\sim N(0,\Sigma)$ and $\Sigma=\diag(\s_1^2,\s_2^2)$.
Covariates are generated from the uniform distribution $U(0,1)$.
We set $\b_0=-3$, $\beta=(2,\ -2,\ 2)'$, $\s_1^2=2$, $\s_2^2=1$, $n_i=10$ and $m=50$.
We use a normal prior $N(0,100I)$ for $\beta$ and $p(\s_k^2)\propto 1/\s_k^2$, $k=1,2$.
Algorithm 1 is used to estimate the posterior mean of the parameters $\t=(\b_0,\b,\s_1^2,\s_2^2)$.
The sequence $a_t$ is set as $1/T$ with $T=10$.
The initial distribution $\pi_0$ is a multivariate $t$ with mean $(0,\ 0,\ 0,\ 0,\ 1,\ 2)$,
scale matrix $3I_6$ and degrees of freedoms $10$.

We first run Algorithm 1 to generate $M=100$ importance samples with $N=10$,
and obtain $\tau_0=7.2\times10^{-3},\tau_1=5.9\times10^{-4}$ and $\bar\gamma^2=17.7$.
This gives $\s_\opt^2=2.6$ and the optimal number of particles (as constant across $\t$) $N_\opt=7$.
We then run Algorithm 1 to generate $M=5000$ importance samples for five values of $N$, $N=1, 7, 10, 20$ and $N=50$.
In order to be able to estimate the variance of the estimator,
Algorithm 1 is run in parallel as described in Remark 2 with $R=20$ batches.

Figure \ref{TNVvsN} plots the time normalized variance in \eqref{eq:TNV} versus $N$.
The TNV is averaged over the posterior mean estimates of the four parameters $\b_0$, $\b$, $\s_1^2$ and $\s_2^2$.
The TNV appears to be minimized at $N=10$, close to the theoretical optimal value $N=7$.
The results suggests that the TNV is weakly sensitive around the optimal value of $N$.
The efficiency decreases linearly when $N$ is higher than the optimal value, whereas
the efficiency can deteriorate exponentially when $N$ is below the optimal.
This phenomenon is also observed in the IS$^2$ method \citep{Tran:2013}.
In practice, it is therefore advisable to use for $N$ a value which is slightly bigger than $N_\opt$.

Table \ref{table1} reports the estimate of the posterior mean when using $N=10$ particles.

\begin{figure*}[h]
\centerline{\includegraphics[width=.5\textwidth,height=.35\textwidth]{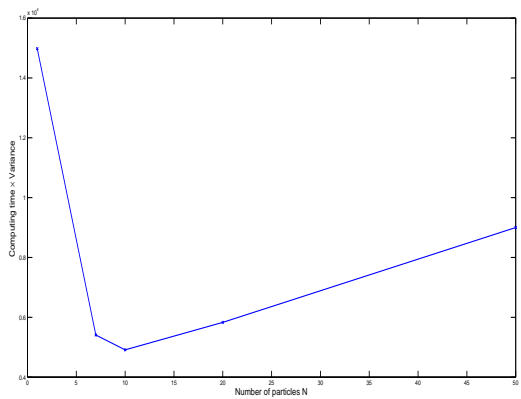}}
\caption{\label{TNVvsN}
Plot of the time normalized variance vs $N$}
\end{figure*}

\begin{table}[h]
\centering
\vskip2mm
{\small
\begin{tabular}{cccc}
\hline\hline %inserts double horizontal lines
	   & True & Mean 	& Std. Dev\\
$\beta_0$  &-3	  & -3.08	&0.40\\
$\b_1$	   &2	  &1.97		&0.08  	\\
$\b_2$	   &-2	  &-1.99	&0.05	\\
$\b_3$	   &2	  &2.02		&0.08\\
$\s_1^2$   &2	  &2.15		&0.68\\ 	 	
$\s_2^2$   &1	  &0.62		&0.07\\ 	 	
\hline\hline
\end{tabular}
}
\caption{Estimate of the posterior mean and the standard deviation}\label{table1}
\end{table}

\subsection{Real Data Example: The SV model}
We analyze the Pound/Dollar data set (\cite{KimShephardChib1998}) using both
the standard SV model and the SV model with leverage effect. For the standard SV
model the measurement equation, for the $t^{th}$ observation $y_t$,
is given by
\beq \label{eq:measurement_SV}
y_t = \exp\left(\frac{h_t}{2}\right)\varepsilon_t;\ t=1,2,\dots,n,
\eeq
where $\varepsilon_t$ follows the standard normal distribution, and $h_t$ is a
latent variable defined for $t=1,2,\dots,n-1$, as
\beq \label{eq:state_SV}
h_{t+1} = \mu\left(1-\phi\right) + \phi h_t + \sigma_{\eta} \eta_t,
\eeq
where $\mu$ is the unconditional mean, $\phi$ is the level
of persistence, $\sigma_{\eta}$ is the scaling parameter, for the latent
process and $\eta_t$ follows the standard normal distribution. The SV model given in \eqref{eq:measurement_SV} and
\eqref{eq:state_SV} is completed by the initial state, given as
\beq \label{eq:initial_state}
h_1\sim N\left(m_1,v_1\right).
\eeq

For the analysis we assume, \emph{a priori}, that $\mu$ is normally distributed,
such that $\mu\sim N(0, 100)$, we assume that the log of variance, $h_t$, is
generated by a stationary process with positive autocorrelation, and where the
persistence parameter,
$\phi$, follows a Beta distribution
where $\phi \sim Be(15, 1.5)$ and we assume $\sigma_\eta$ has an inverted gamma
prior $IG(10, 0.1)$.

To implement Algorithm 1, we use the bootstrap particle filter of
\cite{GordonSalmondSmith1993} to obtain an
unbiased estimate of the log-likelihood. For this specific implementation
$\tau_0=0$, for which \eqref{eq:sigsq_opt} implies that the variance of the
estimated log-likelihood should be $\frac{1}{\tau}$ for the annealed IS scheme that
we implement. In the analysis of the SV model we set $T=15$.
To estimate the optimal number particles, we compute the
log-likelihood one thousand times, evaluated at parameter values that are
typical for the SV model. Specifically, we set $\mu=-0.6$, $\phi=0.98$ and we
set $\sigma_{\eta}=0.16$. This leads us to the conclusion that 24 is the optimal
number of particles for this model and data set. Algorithm 1 also requires us to
set $M$, which we set to $M=1000$.

We also need to specify step (iii) of Algorithm 1, which is the Markov move
step. Specifically, we employ 5 random walk Metropolis Hastings (RWMH) steps for each
Markov move; see \cite{RobertCassela1999} for further details on the RWMH algorithm.
Similarly, to \cite{Chopin2002} (in the sequential Monte Carlo context),
we take the covariance of the RWMH
algorithm as the covariance of the current set of (annealed IS) particles. The
covariance is scaled by the parameter $\alpha$, where $\alpha$ is adjusted at
each move step, based on the acceptance rate of the previous move step.
Specifically, at each move step we update the scale parameter, $\alpha$, such
that

\beq \label{eq:mult_factor}
\alpha \leftarrow MF \times \alpha,
\eeq
where $MF$ denotes a multiplication factor that scales $\alpha$ from the previous
period. The multiplication factor is determined by the acceptance rate ($AR$) from the
previous Markov move step.

\begin{table}[h]
\centering
\vskip2mm
{\small
\begin{tabular}{c|cccccccccc}
\hline %inserts double horizontal lines
Range $AR$ & $[0,0.01)$ & $[0.01, 0.1)$ & $[0.1, 0.15)$ & $[0.15, 0.2)$ & $[0.2, 0.23)$\\
$MF$ & 0.2 & 0.5 & 0.7 & 0.9 & 0.99\\
\hline
Range $AR$	   & $[0.23, 0.25)$ &$[0.25, 0.5)$ & $[0.5, 0.85)$ & $[0.85, 0.99)$&  $[0.99, 1.0)$\\
$MF$ & 1 & 1/0.97 & 1/ 0.8 & 1/0.7 & 1/0.5\\
\hline
\end{tabular}
}
\caption{reports the value of the multiplication factor, $MF$, given the
acceptance rate of the previous Markov move step}
\end{table}

We report in Table \ref{MultFactor} the multiplication factor used in
\eqref{eq:mult_factor}, given the value of the acceptance rate from the previous
period. We find adapting in this fashion works well in all the examples we have
considered so far, including the ones we consider in this paper.

\begin{figure*}[h]
\centerline{\includegraphics[width=.5\textwidth,height=.35\textwidth]{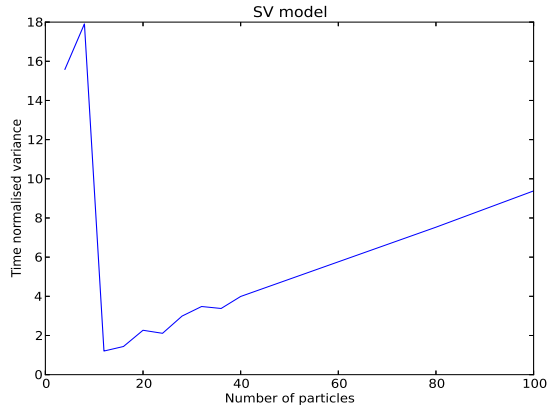}}
\caption{
Plot of the time normalized variance vs $N$ for the SV model}\label{MultFactor}
\end{figure*}

The output for Figure \ref{MultFactor} is produced by running Algorithm 1,
using the method described in Remark 3 with $R=200$ batches. For this part of
the analysis
here we set $a_t=t/n$. As in the simulated example,
the theoretically optimal number of simulated particles, is close the what
is empirically observed as optimal. We also observe that the penalty for using
too few particles, with respect to the time corrected measure of accuracy, can
be much greater than using too many particles. 

\begin{table}[h]
\centering
\vskip2mm
{\small
\begin{tabular}{ccc}
\hline %inserts double horizontal lines
    & Mean & Std. Dev\\
$\mu$ & -0.66  & 0.40 \\
$\phi$ & 0.98 & 0.02\\
$\sigma_\eta$ & 0.17 & 0.04\\
\hline
log ML & -19 & \\
\hline
\end{tabular}
}
\caption{Reports output from the analysis of the SV model on the Pound/Dollar
data set.} \label{PoundDollar}
\end{table}

We use the annealed IS algorithm for estimation of the
SV model on the Pound/Dollar data set.  Here we set $a_t=(t/n)^3$. This ensures
that we move away from the prior very slowly initially, and ensures that the
effective sample size is high as we move over the entire extended state space.
This is important when estimating the marginal likelihood as we require accurate
estimates of the expected value of the log-likelihood across the entire
temperature range, and not just at the end of the estimation process, which is
the case for parameter estimates. The results are reported in Table
\ref{PoundDollar}.  We report both the posterior mean and
standard deviation for each of the parameters. The time taken for estimation is
176 seconds, using the Julia programming language (\cite{JULIA2012}), on a Core
i7 Linux box, with 16 Gigabytes of RAM. Note the code has not been parallelized,
so even better performance could be achieved, with a more highly optimized
implementation.

Unlike the standard SV model, in \eqref{eq:measurement_SV}, \eqref{eq:state_SV}
and \eqref{eq:initial_state}, in which the measurement and state disturbances
are independent, the SV model with leverage effect allows for correlation
between $\varepsilon_t$ and $\eta_t$; see \cite{OmoriChibShephardNakajima2007},
for further details.  Specifically, it is assumed that
\bean
\left(
\begin{array}{c}
\varepsilon_{t}\\
\eta_{t}
\end{array}
\right)\sim N\left(0, \left(\begin{array}{cc}
		1 & \rho \\
		\rho & 1
		\end{array}\right)\right),
\eean
where $\left|\rho\right|<1$. As for the standard SV model we use the bootstrap
particle filter to obtain an estimate of the log-likelihood.   We set the number
of particles to 20, which corresponds to our estimate of the optimal number of
particles. The other algorithmic parameters remain the same as for the standard
SV model.

\begin{table}[h]
\centering
\vskip2mm
{\small
\begin{tabular}{ccc}
\hline %inserts double horizontal lines
    & Mean & Std. Dev\\
$\mu$ & -0.71  & 0.26 \\
$\phi$ & 0.98  & 0.02 \\
$\sigma_\eta$ & 0.17  & 0.04\\
$\rho$ & -0.04 & 0.12 \\
\hline
log ML & -73 & \\
\hline
\end{tabular}
}
\caption{Reports output from the analysis of the SV model with leverage on the Pound/Dollar
data set.} \label{SV_lev}
\end{table}

Output from the analysis of the Pound/Dollar data set using the SV model with
leverage is reported in Table \ref{SV_lev}.  The analysis took 172 seconds on
using the Julia programming language (\cite{JULIA2012}), on a Core
i7 Linux box, with 16 Gigabytes of RAM. There is little evidence of leverage
effect, for the Pound/Dollar data set, based on this analysis. In particular,
the log marginal likelihood strongly favours the standard SV model, which isn't
surprising given that the estimate of $\rho$ is close to zero.

%======================================================================%
\section{Conclusions}\label{Sec:conclusion}
%======================================================================%
We have presented the annealed IS algorithm
for Bayesian inference in models with latent variables.
The proposed AISEL method can be considered as a supplement to existing
Monte Carlo methods for latent variable models, including particle MCMC \citep{Beaumont:2003,Andrieu:2009,Andrieu:2010},
SMC$^2$ \citep{Chopin:2012} and IS$^2$ \citep{Tran:2013}.
The theory and methodology presented in this paper are useful
for Bayesian inference in latent variable models where the posterior distribution is multimodal
and choosing an appropriate proposal density is challenging. An estimate of the
log marginal likelihood is obtained as a byproduct of the estimation procedure.

%======================================================================%
\section*{Acknowledgment}
%======================================================================%
The research was partially supported by Australian Research Council grant DP0667069.

%======================================================================%
\section*{Appendix}
%======================================================================%
\begin{proof}[Proof of Theorem \ref{the:theorem ESS}]
As in Section 2, we can see that the AISEL algorithm is an IS procedure operating on the extended space $(\Theta\otimes\Bbb{R})^T$
with the proposal density
\beqn
\wt g(x^{(1)},...,x^{(T)})=\wt\xi_{a_0}(x^{(1)})K_{\wt\xi_{a_1}}(x^{(1)},x^{(2)})...K_{\wt\xi_{a_{T-1}}}(x^{(T-1)},x^{(T)}).
\eeqn
Note that $\wt\eta_{a_t}(\t,z)=\eta_{a_t}(\t)\exp(a_tz)g_N(z|\t)$.
We have
\bea\label{eq:log w}
\log(\wt w_i)&=&\sum_{t=1}^T\left(\log\wt\eta_{a_t}(x^{(t)})-\log\wt\eta_{a_{t-1}}(x^{(t)})\right)\notag\\
&=&\sum_{t=1}^T\left(\log\eta_{a_t}(\t^{(t)})-\log\eta_{a_{t-1}}(\t^{(t)})\right)+\sum_{t=1}^T(a_t-a_{t-1})z_t\notag\\
&=&\log(w_i)+\sum_{t=1}^T(a_t-a_{t-1})z_t.
\eea
Denote $\wt\t=(\t^{(1)},...,\t^{(T)})$ and $\wt z=(z^{(1)},...,z^{(T)})$.
Under Assumption 4,
\beqn
g_\t(\wt\t)=g_\t(\t^{(1)},...,\t^{(T)})=\xi_{a_0}(\t^{(1)})\xi_{a_1}(\t^{(2)})...\xi_{a_{T-1}}(\t^{(T)}).
\eeqn
and
\bean
\wt g(x^{(1)},...,x^{(T)})&=&\wt\xi_{a_0}(x^{(1)})\wt\xi_{a_1}(x^{(2)})...\wt\xi_{a_{T-1}}(x^{(T)}).
\eean
By Assumption 3,
\beqn
\wt\xi_{a_t}(\t,z)=\xi_{a_t}(\t)g_t(z)\quad\text{with}\quad g_t(z)=\frac{1}{C_t}e^{a_tz}g(z|\s),\quad C_t=\exp(-\frac12a_t\s^2+\frac12a_t^2\s^2).
\eeqn
Hence
\beqn
\wt g(x^{(1)},...,x^{(T)})=\prod_{t=1}^T\xi_{a_{t-1}}(\t^{(t)})\prod_{t=1}^Tg_{t-1}(z^{(t)})=g_\t(\wt\t)g_z(\wt z)
\eeqn
with
\beqn
g_z(\wt z)=\prod_{t=1}^Tg_{t-1}(z^{(t)}).
\eeqn
From \eqref{eq:log w}, we have
\bean
\E_{\wt g}[\wt w_i^2]&=&\int w_i^2(\wt\t)\prod_{t=1}^Te^{2(a_t-a_{t-1})z^{(t)}}g_z(\wt z)g_\t(\wt\t)d\wt zd\wt\t\\
&=&\int w_i^2(\wt\t)g_\t(\wt\t)d\wt\t\times \prod_{t=1}^T\int e^{2(a_t-a_{t-1})z^{(t)}}g_{t-1}(z^{(t)})dz^{(t)}\\
&=&\E_g[w_i^2]\prod_{t=1}^T\frac1{C_{t-1}}\exp\left(-\frac12{(2a_t-a_{t-1})\s^2}+\frac12{(2a_t-a_{t-1})^2\s^2}\right)\\
&=&\E_g[w_i^2]\exp\left(\s^2\sum_{t=1}^T(a_t-a_{t-1})(2a_t-1)\right)\\
&=&\E_g[w_i^2]\exp\left(\tau\s^2\right)
\eean
with $\tau=\sum_{t=1}^T(a_t-a_{t-1})(2a_t-1)$.
Note that $\E_{\wt g}[\wt w_i]=\E_g[w_i]$ by the unbiasedness property of IS.
Hence,
\bea\label{eq:CV}
1+\CV_{\wt g}(\wt w_i)&=&1+\frac{\Var_{\wt g}(\wt w_i)}{(\E_{\wt g}[\wt w_i])^2}=\frac{\E_{\wt g}[\wt w_i^2]}{(\E_{\wt g}[\wt w_i])^2}=
\frac{\E_{g}[w_i^2]}{(\E_{g}[w_i])^2}\exp(\tau\s^2)\notag\\
&=&\left(1+\CV_{g}(w_i)\right)\exp(\tau\s^2).
\eea
So
\beq\label{e:defineESS}
\ESS_\AISEL = \frac{M}{1+\CV_{\wt g}(\wt w_i)}=\frac{M}{1+\CV_{g}(w_i)}\exp(-\tau\s^2)=\exp(-\tau\s^2)\ESS_\AIS.
\eeq
Note that
\bean
\tau&=&\sum_{t=1}^T(2a_t^2-2a_ta_{t-1})-\sum_{t=1}^T(a_t-a_{t-1})>\sum_{t=1}^T(2a_t^2-a_t^2-a_{t-1}^2)-1=0.
\eean

\end{proof}
\begin{proof}[Proof of Proposition \ref{proposition}]
The proof follows \cite{Friel:2008} who prove the result in the case $\pi_0(\t)=p(\t)$
and the likelihood $p(y|\t)$ is analytically available.
Let
\beqn
\zeta(s)=\int \wt\eta_s(\t,z)d\t dz,\ 0\leq s\leq 1.
\eeqn
Then, $\zeta(0)=1$ and by Assumption 1, $\zeta(1)=p(y)$. 
Note that $\wt\xi_s(\t,z)=\wt\eta_s(\t,z)/\zeta(s)$.
\bean
\frac{d\zeta(s)}{ds}&=&\frac{d}{ds}\int\pi_0(\t)\left(\frac{p(\t)p(y|\t)e^z}{\pi_0(\t)}\right)^sg_N(z|\t)d\t d z\\
&=&\int\pi_0(\t)\left(\frac{p(\t)p(y|\t)e^z}{\pi_0(\t)}\right)^s\left[\log\frac{p(\t)p(y|\t)e^z}{\pi_0(\t)}\right]g_N(z|\t)d\t d z\\
&=&\int \wt\eta_s(\t,z)\left[\log\frac{p(\t)p(y|\t)e^z}{\pi_0(\t)}\right]d\t d z.
\eean
Hence,
\bean
\frac{d\log\zeta(s)}{ds}&=&\frac{1}{\zeta(s)}\frac{d\zeta(s)}{ds}\\
&=&\int \wt\xi_s(\t,z)\left[\log\frac{p(\t)p(y|\t)e^z}{\pi_0(\t)}\right]d\t d z\\
&=&\E_{(\t,z)\sim\wt\xi_s}\left[\log\frac{p(\t)p(y|\t)e^z}{\pi_0(\t)}\right].
\eean
So
\beqn
\int_0^1\E_{(\t,z)\sim\wt\xi_s}\left[\log\frac{p(\t)p(y|\t)e^z}{\pi_0(\t)}\right]ds = \log\zeta(1)-\log\zeta(0)=\log p(y).
\eeqn
\end{proof}

\bibliographystyle{apalike}
\bibliography{references}

\end{document}